\documentclass[10pt,twocolumn,a4paper]{IEEEtran}

\makeatletter
\let\NAT@parse\undefined
\makeatother

\newif\ifnotsergio
\newif\ifsergio
\notsergiotrue
\ifsergio
\usepackage[linkcolor={0 0 0}, pdfborder={0 0 0}, linktocpage=true,
hyperindex=true,bookmarks=true, bookmarksnumbered=true]{hyperref}
\fi

\usepackage[english]{babel}
\usepackage[applemac]{inputenc}%utf8

\usepackage{flushend}
\usepackage{lscape}
\usepackage{bigints}
\usepackage{graphicx} % include this line if your document contains figures
\usepackage{cite}
\usepackage{psfrag}
\usepackage{amsfonts}
\usepackage{amssymb}
\usepackage{amsmath}
\usepackage{graphicx}
\usepackage{psfrag}
\usepackage{ae,epsfig}
\usepackage{verbatim} % allows for the use of \begin{comment}
\usepackage{enumerate} %easier custom labels for enumerate
\usepackage{epstopdf}
\usepackage{tikz}
\usetikzlibrary{shapes,arrows}
\usetikzlibrary{positioning,fit}
\usepackage{color}
\usetikzlibrary{positioning}
\usepackage{caption}
\usepackage{subcaption}
\usepackage{tikz}
\usetikzlibrary{shapes,arrows}
\usetikzlibrary{positioning,fit}
\usepackage{bm}
\usepackage{bbm}
\usetikzlibrary{calc}
\usetikzlibrary{matrix}

\usepackage{color}

\newcounter{teocount}
\newcounter{propcount}

\newcounter{remcount}

\newcounter{excount}

\newtheorem{remm}[remcount]{Remark}

\newtheorem{proposition}[propcount]{Proposition}
\newtheorem{theorem}[teocount]{Theorem}

\newtheorem{ex}[excount]{Example}

\newenvironment{proof}{{\em Proof. }}{\hfill \hspace*{1pt} \hfill $\blacksquare$}
\newenvironment{remark}{\begin{remm}\rm }{\hfill \hspace*{1pt} \hfill
$\star$\end{remm}}

\begin{document}
\author{Ross Drummond, Luis D. Couto and Dong Zhang
\thanks{ Ross Drummond is with the Department of Engineering Science, University of Oxford, 17 Parks Road, OX1 3PJ, Oxford, United Kingdom. Email: ross.drummond@eng.ox.ac.uk.

 Luis D. Couto is with the Department of Control Engineering and System Analysis, Universit\'{e} libre de Bruxelles, Brussels, B-1050, Belgium. Email: lcoutome@ulb.ac.be.

Dong Zhang is with the Department of Mechanical Engineering, Carnegie Mellon University, 5000 Forbes Avenue, Pittsburgh, PA 15312, United States. Email: dongzhr@cmu.edu.

}
}

\IEEEoverridecommandlockouts

\title{Resolving Kirchhoff's laws for state-estimator design of Li-ion battery packs connected in parallel}

\maketitle
\thispagestyle{plain}
\pagestyle{plain}

\begin{abstract}

A state-space model for Li-ion battery packs with parallel connected cells is introduced. 
The key feature of this model is an explicit solution to Kirchhoff's laws for parallel connected packs, which expresses the branch currents directly in terms of the model's states, applied current and cell resistances.
This avoids the need to solve these equations numerically. 
To illustrate the potential of the proposed model for pack-level control and estimation, a state-estimator is introduced for the nonlinear parallel pack model.
By exploiting the system structure seen in the solution to Kirchhoff's laws, algebraic conditions for the observer gains are obtained that guarantee convergence of the estimator's error. Error convergence is demonstrated through an argument based upon Aizerman's conjecture. 
It is hoped that the insight brought by this model formulation will allow the wealth of results developed for series connected packs to be applied to those with parallel connections. 

\end{abstract}


%%%%%%%%%%%%%%%%%%%%%%%%%%%%%%%%%%%%%%%%%%%%%%%%%%%%%
\section*{Introduction}

To address ever increasing energy and power demands, Li-ion battery pack sizes are growing rapidly, especially for large-scale applications such as electric vehicles and grid storage. In some parts of the world, it is now common to see electric vehicles powered by thousands of cells, like the Tesla Model S \cite{brand2016current},  and large batteries, like the planned 50 MW battery to be run near Oxford by Pivot Power \cite{pivot}, are now coming online to support the grid. The sheer number of cells in these large battery packs introduces several challenges that need to be overcome, especially with the design of the battery management system (BMS). The BMS is predominantly responsible for estimating the state of charge and health of the pack, however, as pack sizes continue to grow, ensuring that the BMS algorithms remain both accurate and scalable enough to be implemented on embedded hardware is becoming ever more challenging.  

Battery models are the foundations for any advanced BMS and to perform at its best, it is desirable for the BMS to have information about every cell in the pack. This has motivated significant efforts to develop models for whole battery packs. However, whilst most large battery packs used in practice are mixtures of both parallel and series connections, most studies on pack level modelling and BMS design are restricted to just series connections, for example \cite{zhong2014method,Lin-2014}. Focusing explicitly on series connected cells greatly simplifies the problem, as every cell in series is charged with the same current, but neglects the diverse spectrum of pack configurations seen in practice. 

Whilst including parallel connections into the pack can bring many benefits, such as increased reliability \cite{brand2016current} and natural self-balancing \cite{zhong2014method}, modelling and supervising parallel connected cells has proven to be more challenging than cells in series. This  is primarily because the branch currents charging each parallel branch have to be computed at each time instant in the models. The branch currents are obtained by computing solutions to Kirchhoff's laws, which makes the resulting pack models differential algebraic equation models (DAEs). DAEs models can be significantly more complex than those described by ordinary differential equations (ODEs), so most studies on parallel packs numerically compute solutions to Kirchhoff's laws before projecting the state of the index 1 DAE down into an ODE. Examples of this approach include the iterative scheme of \cite{diao2019management}, the frequency domain  approximations  of \cite{chang2019correlations} and the numerical matrix inversion methods of studies like \cite{D_L} and \cite{bruen2016modelling} which was augmented with a thermal model in \cite{hosseinzadeh2019combined}. In contrast, this work obtains an ODE model by providing an analytical solution to Kirchhoff's laws for $n$-cells connected parallel. Thus, the main result of this work can be thought of as providing an analytical solution, in terms of the various cell resistances, to the branch current equations defined by Kirchhoff's laws, in place of the numerical solutions in benchmark studies like (20) of \cite{D_L} and (15) of \cite{bruen2016modelling}. This approach follows  along a recent direction in battery pack modelling, including the cell merging approach of \cite{fan2020simplified}, and generalises similar efforts like \cite{fill2018current, hofmann2018dynamics} by relaxing some of the restrictive modelling assumptions, like the linearity of the open cell voltage \cite{fill2018current}, as well as providing a more involved  model formulation that additionally includes the important state-of-charge dynamics than \cite{hofmann2018dynamics}. With a state-space formulation for the parallel connected Li-ion battery pack in hand, the state-estimator design problem can then be addressed, with simple gain conditions given in  Section \ref{sec:obs}.

\emph{Contribution:} To be specific, the main contribution of this paper is to introduce a state-space model for parallel connected packs that is fully described by an ordinary differential equation explicitly parametersied by the various resistances and capacitances of the pack's cells. To achieve this, an analytic solution to the algebraic equation of Kirchhoff's laws is stated (see Section \ref{sec:alg}). With this equation in hand, the parallel pack model can then be condensed into a state-space form with an appealing structure that can be exploited for analysis. To illustrate this point, a new state estimator for this pack model is introduced in Section \ref{sec:obs}, whose main benefit over existing approaches is that checking asymptotic convergence of the estimator  error for the nonlinear system is simple, as the conditions for convergence 
 are algebraic. 
 
 State estimators are key components of battery management systems, but, the estimator design problem for parallel packs has received significantly less attention than that for cells in series. There are two main reasons for this: 1. It is widely assumed that cells in parallel have the same state-of-charge because of this setup's natural self-balancing \cite{zhong2014method} but the simulations of \cite{D_L} suggest this may not always hold; 2. The need to resolve the branch currents makes parallel pack models more complex to analyse. The results proposed here are directed at this second issue, with the analytic expression for the branch currents bringing insight into the model structure that can be exploited. 
 
% One of the main reasons behind this research gap is the added complexity of parallel packs models (with this complexity coming from the need to compute the branch currents) \cite{D_L}. By providing an analytic expression for these branch currents and methods to tune the estimator gains (Propositions 1 and 2), this research aims to resolve this issue. 

The results presented here are in many ways an extension of the recent results of \cite{D_L} from some of the authors. In \cite{D_L}, an observer was designed for a DAE model of a parallel connected pack but the DAE element of this model introduced severe complexity into the analysis of the model's vector field and the conditions guaranteeing convergence of the observer error. This paper resolves these issues by exploiting the analytic expression for the parallel branch current. It is hoped that the analysis presented in this paper will lead to new results in other applications where parallel pack models are used, for example in determining the weakest cells in the packs, detecting thermal runaway and enabling whole pack state-estimators for large Li-ion battery packs.

\begin{figure}
\centering
\includegraphics[width=0.4\textwidth]{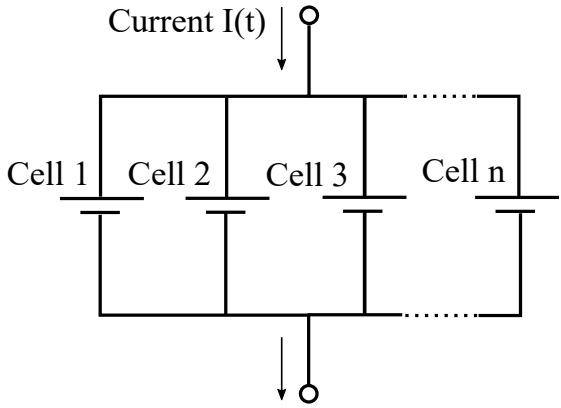}
        \caption{A Li-ion battery pack containing $n$ cells connected in parallel.} 
        \label{fig:pack}
\end{figure}

\begin{figure}
\centering
%\graphicspath{{./Figures/}}
\includegraphics[width=0.4\textwidth]{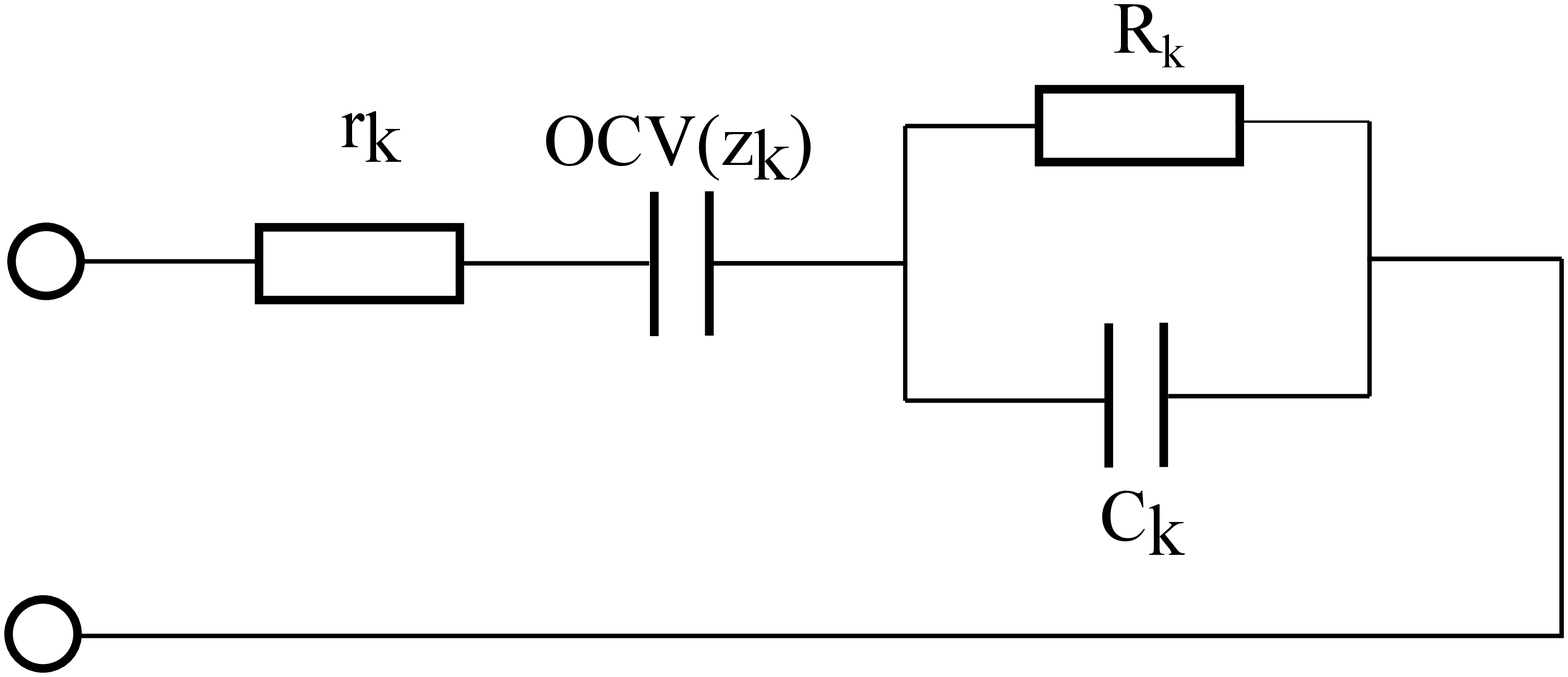}
        \caption{Equivalent circuit model for the battery dynamics. Here, $r_k$ is the $k^\text{th}$ cell's resistance,  OCV$(z_k)$ is its open circuit voltage and $(R_k,\,C_k)$ denote an RC-pair. }
        \label{fig:circ}
\end{figure}

%----------------------------------------------------------------------------------
\emph{Notation: }
If a square matrix $A$ of dimension $n$ is positive definite then $A \in \mathbb{S}^n_{\succ_0}$ and if it is negative definite then $A \in \mathbb{S}^n_{\prec_0}$.  More generally, if a matrix $A$ is negative-definite then $A \prec 0 $. If $A$ is a non-negative diagonal matrix then  $A \in \mathbb{D}^n_{+}$. The identity of dimension $n$ is denoted $I_n$. A signal $x(t)$ is said to be in the Hilbert space $x \in \mathcal{L}_2$ if the norm
\begin{align}
    \|x\|_2 = \sqrt{\int^{\infty}_{0} x(t)^2~dt}
\end{align}
is bounded. 
%----------------------------------------------------------------------------------
\section{DAE model for a parallel connected pack}\label{sec:DAE}
In this section, the equations of a DAE model for Li-ion batteries connected in parallel are described. In Section \ref{sec:alg}, this DAE model is converted into an ODE by resolving the underlying algebraic equation for the branch currents.

\subsection{Parallel pack model equations}
Figure \ref{fig:pack} shows the set-up of the parallel connected Li-ion battery pack that is to be modelled. Each cell is assumed to be described by the equivalent circuit model of Figure \ref{fig:circ}, composed of a capacitor (for the state-of-charge) and an RC pair (generally associated with solid-state diffusion in the active material particles). The dynamics of the $k^{\text{th}}$ cell in the pack with this circuit model are  
\begin{subequations}\label{circuit}\begin{align}
\dot{x}_k(t)  & = \bar{A}_kx_k(t) + \bar{B}_k i_k(t), \quad k = 1, \dots, n,\\
v_k(t)  & = w_k(t) + \text{OCV}(z_k(t))+ r_ki_k(t),
\end{align}\end{subequations}
where $x_k = [z_k, w_k]$ is the state of the system, $z_k$ is the state-of-charge of each cell and $w_k$ is the relaxation voltage of the capacitor in the $k^{\text{th}}$ RC pair. The current going into each parallel branch is $i_k(t)$. Each cell's voltage $v_k(t)$ is a function of the relaxation voltage $w_k$, the open circuit voltage OCV($z_k$) and the resistance $r_k$. Because the cells are connected in parallel, each cell's voltage is the same $v_k(t) = v(t), \, \forall k = 1, \dots, n$. The state space matrices in \eqref{circuit} are 
\begin{align}
\bar{A}_k = \begin{bmatrix} 0 & 0 \\ 0 & -\frac{1}{R_kC_k} \end{bmatrix}, \quad \bar{B}_k = \begin{bmatrix} \frac{1}{Q_k} \\ \frac{1}{C_k}\end{bmatrix}.
\end{align}
where $Q_k$ is the battery capacity and $(R_k,C_k)$ represent the RC pair.

What remains is to compute the branch current $i_k$ going into each cell. This is achieved by applying Kirchhoff's laws. Namely, Kirchhoff's voltage law implies
\begin{subequations}\label{Kirchhoff}\begin{align}
\text{OCV}(z_j) + w_j + r_{j}i_j= \text{OCV}(z_k) + w_k + r_k i_k, \nonumber \\
j,k \in \{ 1, 2, \dots, n\},
\end{align}
and the current law states that the sum of the currents going into each branch $i_k(t)$ equals the pack current $I(t)$
\begin{align}
\sum^n_{k = 1}i_k(t) = I(t).
\end{align}\end{subequations}

\subsection{Differential algebraic equation parallel pack model}
When combined, the dynamic circuit equations \eqref{circuit} and the algebraic equations for Kirchhoff's laws \eqref{Kirchhoff}  can be collected into a single DAE system \cite{D_L},
\begin{align}\label{sys_DAE}
\begin{bmatrix}I_{2n} & 0 \\ 0 & 0 \end{bmatrix}
\begin{bmatrix}\dot{x} \\ \dot{i} \end{bmatrix}
=
\begin{bmatrix}A_{11} & A_{12} \\ A_{21} & A_{22} \end{bmatrix}
\begin{bmatrix}{x} \\ {i} \end{bmatrix}
+
\begin{bmatrix}0 \\ \phi(t) \end{bmatrix}
\end{align}
where $A_{11} = \text{diag} (\bar{A}_1, \, \bar{A}_2, \, \dots, \, \bar{A}_n)$, $A_{12} = \text{diag} (\bar{B}_1, \, \bar{B}_2, \, \dots, \, \bar{B}_n)$,
\begin{subequations}\begin{align}
A_{22}   & = \begin{bmatrix}
r_1  & -r_2 & 0 & \dots & 0 \\
r_1 & 0 & -r_3 & \ddots & \vdots \\
\vdots & \vdots & \ddots & \ddots & 0 \\
r_1 & 0 & \dots & 0 & -r_n \\
1 & 1 & 1 & \dots & 1
\end{bmatrix}\label{A22}
\\
A_{21}  & =  \begin{bmatrix} 0 & 1 & S & 0 & \dots & 0 \\ 
0 & 1 & 0 & S & \ddots & \vdots \\
\vdots & \vdots & \vdots & \ddots  & \ddots & 0\\
0 & 1 & 0 & \dots  & 0& S \\ 
0 & 0 & 0 & 0 & \dots & 0 \end{bmatrix}, \quad S = \begin{bmatrix} 0 & -1 \end{bmatrix},
\end{align}
and
\begin{align}
\phi(t) = \begin{bmatrix}\text{OCV}(z_1)-\text{OCV}(z_2) \\  \vdots \\ \text{OCV}(z_1)-\text{OCV}(z_n) \\ I(t) \end{bmatrix}.
\end{align}\end{subequations}
The variables with a time derivative $x = [x_1, \, \dots, \, x_n]^T$ are known as the differential or state-space variables whilst the current vector $i = [i_1, \, \dots, \, i_n]^T$ is  the model's algebraic variable.  Due to the linearity of Kirchhoff's laws, these currents can be obtained directly from
\begin{align}\label{i_exp}
i(t) = -{A_{22}}^{-1}A_{21}x(t)  -{A_{22}}^{-1}\phi(t),
\end{align}
given that the matrix $A_{22}^{-1}$ is invertible (as shown in Section \ref{sec:alg}). Substituting the expression for the currents \eqref{i_exp} into the DAE model \eqref{sys_DAE} reduces it to an ODE
\begin{align}\label{sys_ODE}
\dot{x}(t) = (A_{11}-A_{12}{A_{22}}^{-1}A_{21})x(t) -A_{12}{A_{22}}^{-1}\phi(t).
\end{align}

This is a standard model for parallel connected Li-ion battery packs, but, no analytic expression has previously been obtained for the matrices $A_{12}{A_{22}}^{-1}A_{21}$ nor $A_{12}{A_{22}}^{-1}$. In the next section, these expressions are given by writing out an explicit solution for the inverse of the $A_{22}$ matrix.

\section{Resolving the algebraic equation}\label{sec:alg}

The main result of this paper are contained in this section where an algebraic solution for the current going into each branch of the parallel circuit is provided. In this way, the the parallel connected pack model of Section \ref{sec:DAE} can be fully characterised. 

\subsection{The matrix inverse ${A_{22}}^{-1}$}
The main stumbling block behind resolving the algebraic equation for the currents is determining the matrix inverse ${A_{22}}^{-1}$. Thankfully, because this matrix contains a nice structure similar to an atomic matrix, its inverse can be readily computed. 

\begin{theorem}\label{thm:A22}
Consider the matrix $A_{22}$ in \eqref{A22} with $r_i>0$ for $i = 1, \, \dots, \, n$. Then ${A_{22}}^{-1} = M$ where $M$ is a matrix composed of elements $m_{j,k}$ satisfying
 \begin{subequations}\begin{align}
m_{j,n} &=  \frac{1}{\displaystyle \sum^n_{\ell = 1}\frac{r_j}{r_{\ell}}}, \quad j = 1,\, \dots,\, n,\label{m_in}
\\
m_{\ell,j}  & =  \frac{1}{r_{\ell}r_{j+1}}\left(\sum^n_{k = 1}\frac{1}{r_k}\right)^{-1}, ~ j = 1, \dots, n-1, \,  \\  & \qquad \qquad \qquad \quad \quad \quad \quad \quad \quad \&  \,\ell  \neq j+1, \nonumber
\\
m_{j+1,j}  &
= 
\frac{1}{{r_{j+1}}^2}\left(\sum^n_{k = 1}\frac{1}{r_k}\right)^{-1} -\frac{1}{r_{j+1}}.
\end{align}\end{subequations}
\end{theorem}
\begin{proof}
The problem can be cast as finding the unique solution to 
\begin{align}
A_{22}M = I_n,
\end{align}
or, in an expanded form,
\begin{align}\label{big_mat}
 & \begin{bmatrix}
r_1  & -r_2 & 0 & \dots & 0 \\
r_1 & 0 & -r_3 & \ddots & \vdots \\
\vdots & \vdots & \ddots & \ddots & 0 \\
r_1 & 0 & \dots & 0 & -r_n \\
1 & 1 & \dots & 1 & 1
\end{bmatrix}
\\
 & \begin{bmatrix}
m_{1,1} & m_{1,2} &\dots & \dots & m_{1,n} \\
m_{2,1} & m_{2,2} &\dots & \dots & m_{2,n} \\
\vdots & \vdots &\vdots & \vdots & \vdots \\
m_{n-1,1} & m_{n-1,2}&\dots & \dots & m_{n-1,n}\\
m_{n,1} & m_{n,2} &\dots &\dots & m_{n,n} \\
\end{bmatrix}\\
 & = 
\begin{bmatrix}
1  &0 &  \dots &0 & 0 \\ 0 & 1 & \ddots & \ddots & 0 \\
\vdots & \ddots & \ddots & \ddots & \vdots \\
0 & \ddots & \ddots & \ddots & 0 \\
0 & 0 & \dots & 0 & 1
\end{bmatrix}
\end{align}
where $m_{j,k}$ are the elements of  $M$. 

Multiplying through by the row of 1's in $A_{22}$ gives the following relations for the column sums of $M$
\begin{subequations}\begin{align}
\sum^n_{j = 1} m_{j,k} &= 0, \quad \forall k \neq n, \label{sum_other} \\
\sum^n_{j = 1} m_{j,n} &= 1,\label{sum_i}
\end{align}\end{subequations}
and, similarly, multiplying through by the remaining rows implies
\begin{subequations}\label{resistances}\begin{align}
r_1m_{1,j}- r_km_{k,j}  & = 0,\quad j = 1, \dots, n,     \,j  \neq k-1, \\
r_1m_{1,k-1}- r_km_{k,k-1}  & = 1, \quad  k = 1, \, 2, \, \dots, \, n.
\end{align}\end{subequations}
Subtracting the various expressions in \eqref{resistances} from each other gives
\begin{subequations}\begin{align}
m_{\ell,j}  & =  \frac{r_k}{r_{\ell}}m_{k,j} , ~j = 1, \dots, n, \,    \,k  \neq j+1 , \,\ell  \neq j+1,~ \nonumber
\\ & \qquad \qquad  \quad \&~k = 1, \dots, \, n, \ell = 1, \dots, n, \label{res1} \\
m_{\ell,k-1}  & = \frac{1+ r_km_{k,k-1} }{r_{\ell}}, k = 2, \dots, \, n, \ell = 1, \dots, n. \label{res2}
\end{align}\end{subequations}

From these relations, the $n^{\text{th}}$ column of $M$ can be extracted. Starting from \eqref{sum_i} and substituting in \eqref{res1} means
\begin{align}
\sum^n_{\ell = 1}\frac{r_k}{r_{\ell}}m_{k,n} &= 1.
\end{align}
Fixing $k = n$ gives
\begin{align}
\sum^n_{\ell = 1}\frac{r_n}{r_{\ell}}m_{n,n} &= 1,
\end{align}
so 
\begin{align}
m_{n,n} &= \frac{1}{\displaystyle\sum^n_{\ell = 1}\frac{r_n}{r_{\ell}}}.
\end{align}
Using \eqref{res1}, the rest of the $n^{\text{th}}$ column's elements $m_{{k},n}$ for $k = 1, \dots, n-1$ can be computed
\begin{align}
m_{j,n} &=  \frac{r_n}{r_{j}}\frac{1}{\displaystyle \sum^n_{\ell = 1}\frac{r_n}{r_{\ell}}}
= \frac{1}{\displaystyle \sum^n_{\ell = 1}\frac{r_j}{r_{\ell}}}.\label{m_in}
\end{align}

To compute the remaining elements of $M$, it is noted that \eqref{resistances}
implies
\begin{subequations}\label{some_eqns}\begin{align}
m_{k,j}  & = \frac{r_1}{r_k}m_{1,j},\quad  j = 1, \dots, n, \& \,j  \neq k-1, \label{prv_one} \\
m_{k,k-1}  & = \frac{r_1m_{1,k-1}-1}{r_k}, \quad  k = 1, \, 2, \, \dots, \, n\label{last_one}
\end{align}\end{subequations}
where \eqref{last_one} can be re-written as
\begin{align}
\label{anothereq}
m_{k+1,k}  & = \frac{r_1m_{1,k}-1}{r_{k+1}}.
\end{align}
Substituting expressions \eqref{prv_one} and \eqref{anothereq} into \eqref{sum_other}, gives
\begin{align}
& \sum^n_{\ell = 1, \ell \neq k+1}\frac{r_1}{r_\ell}m_{1,k} + \frac{r_1m_{1,k}-1}{r_{k+1}}= 0, \quad k = 1, \dots, n-1.
\end{align}
In other words,
\begin{align}
 m_{1,k} = \left(\sum^n_{\ell = 1}\frac{r_1}{r_\ell}\right)^{-1} \frac{1}{r_{k+1}}, \quad k = 1, \dots, n-1.
\end{align}
The remaining elements of $M$ are then obtained from \eqref{some_eqns} 
 \begin{subequations}\begin{align}
m_{\ell,j}  & = \frac{r_1}{r_{\ell}}\left(\sum^n_{k = 1}\frac{r_1}{r_k}\right)^{-1} \frac{1}{r_{j+1}} 
\\ & 
=  \frac{1}{r_{\ell}r_{j+1}}\left(\sum^n_{k = 1}\frac{1}{r_k}\right)^{-1}, \quad j = 1, \dots, n \,  \&  \,\ell  \neq j+1, 
\\
m_{j+1,j}  & = \frac{r_1}{r_{j+1}}\left(\sum^n_{k = 1}\frac{r_1}{r_k}\right)^{-1} \frac{1}{r_{j+1}}-\frac{1}{r_{j+1}}
\\ & 
= 
\frac{1}{{r_{j+1}}^2}\left(\sum^n_{k = 1}\frac{1}{r_k}\right)^{-1} -\frac{1}{r_{j+1}}.
\end{align}\end{subequations}
\end{proof}

%\begin{figure*}\begin{subequations}\begin{align}
%A_{12}{A_{22}}^{-1} A_{21} = 
%\begin{bmatrix} 0 & \bar{B}_1\sum_{i = 1}^{n-1}m_{1,i} &  \bar{B}_1\begin{bmatrix} 0  & -m_{11}\end{bmatrix} &  \bar{B}_1\begin{bmatrix} 0  & -m_{12}\end{bmatrix} & \dots &  \bar{B}_1\begin{bmatrix} 0  & -m_{1n-1}\end{bmatrix}
% \\ 
%0 &  \bar{B}_2\sum_{i = 1}^{n-1}m_{2,i} & \bar{B}_2\begin{bmatrix} 0  & -m_{21}\end{bmatrix} & \bar{B}_2\begin{bmatrix} 0  & -m_{22}\end{bmatrix} & \dots & \bar{B}_2\begin{bmatrix} 0  & -m_{2n-1}\end{bmatrix} 
%\\
%\vdots & \vdots & \vdots & \ddots &  \ddots & \vdots 
%\\
%0 &\bar{B}_{n-1}\sum_{i = 1}^{n-1}m_{n-1,i} & \bar{B}_{n-1}\begin{bmatrix} 0  & -m_{n-11}\end{bmatrix} & \bar{B}_{n-1}\begin{bmatrix} 0  & -m_{n-12}\end{bmatrix}  & \dots & \bar{B}_{n-1}\begin{bmatrix} 0  & -m_{n-1,n-1}\end{bmatrix} 
%\\ 
%0 & \bar{B}_{n}\sum_{i = 1}^{n-1}m_{n,i} & \bar{B}_{n}\begin{bmatrix} 0  & -m_{n1}\end{bmatrix} & \bar{B}_{n}\begin{bmatrix} 0  & -m_{n2}\end{bmatrix} & \dots & \bar{B}_{n}\begin{bmatrix} 0  & -m_{n,n-1}\end{bmatrix} \end{bmatrix},
%\end{align}
%\begin{align}
%A_{12} {A_{22}}^{-1}\phi(t)   &= 
% \begin{bmatrix}\bar{B}_1m_{11} &\bar{B}_1m_{12} & \dots &\bar{B}_1m_{1n} \\
%\bar{B}_2m_{21} & \bar{B}_2m_{22},   & \ddots & \vdots 
%\\ \vdots & \ddots & \ddots & \bar{B}_{n-1}m_{n-1n} 
%\\
%\bar{B}_nm_{n1} & \dots & \bar{B}_nm_{nn} & \bar{B}_nm_{nn}
% \end{bmatrix}\begin{bmatrix}\text{OCV}(z_1)-\text{OCV}(z_2) \\  \vdots \\ \text{OCV}(z_1)-\text{OCV}(z_n) \\ I(t) \end{bmatrix}.
%\end{align}\end{subequations}
%\hrulefill
%\end{figure*}

\subsection{State-space model}\label{sec:ODE}
With the matrix ${A_{22}}^{-1}$ inverse defined, an explicit solution for the ODE model \eqref{sys_ODE} can be stated. To arrive at this statement, several matrices and vectors first have to be established. Defining the vector of open-circuit voltages as
\begin{align}
 \text{OCV}(z(t))  = \begin{bmatrix}\text{OCV}(z_1(t))  \\ \vdots \\ \text{OCV}(z_n(t)) \end{bmatrix},
\end{align}
then the solution to the branch current equation \eqref{i_exp} can be expressed as
\begin{align}\label{i_exp2}
i(t) = \Pi_{v}(\text{OCV}(z(t))+w(t)) + \Pi_I I(t)
\end{align}
where
\begin{subequations}\begin{align}
\Pi_v
  & = 
-\begin{bmatrix}\sum_{i = 1}^{n-1}m_{1,i} &  -m_{1,1} &  -m_{1,2} & \dots &  -m_{1,n-1}
\\
\vdots & \vdots & \vdots & \vdots &  \vdots 
\\ 
\sum_{i = 1}^{n-1}m_{i,n} & -m_{n,1} & -m_{n,2} & \dots & -m_{n,n-1}\end{bmatrix},
\\
\Pi_I  &= -
 \begin{bmatrix}m_{1,n} &
 \dots 
& m_{n-1,n} 
 & 
 m_{n,n}
 \end{bmatrix}^T.\label{Pi_I}
\end{align}\end{subequations}
 Next, the vector of  concatenated voltages is defined
\begin{align}
\check{v}(t)  =\bm{1}_nv(t) =  \begin{bmatrix}w_1(t) + \text{OCV}(z_1(t))+ r_1i_1(t) \\ \vdots \\ w_n(t) + \text{OCV}(z_n(t))  + r_ni_n(t)\end{bmatrix}.
\end{align}
Using the substitution \eqref{i_exp2}, this voltage vector of repeating elements can be formulated as
%\begin{subequations}
\begin{align}
\label{yeq}
\check{v}(t)   &  = Cx(t) + D_{\text{OCV}}\text{OCV}(z(t)) + D_II(t).
\end{align}
%\end{subequations}
where $C = W + \bm{r}\Pi_v {W}$ with $\bm{r} = \text{diag}(r_k), \, k = 1, \, \dots, \, n$ and $W \in \mathbb{R}^{n \times 2n}$ is a matrix full of zeros apart from $W_{k,j} = 1$ if $j = 2k$ for $k = 1, \, \dots, \, n$ (more explicitly, $w = Wx$).  Also,  $D_{\text{OCV}}~=~I_n~+~\bm{r}\Pi_v$  and $D_{I}~=~\bm{r}\Pi_{I}$.  

With these matrices defined, the dynamics of  \eqref{sys_ODE} can be written as 
\begin{align}
\label{xeq}
\dot{x}  & = Ax + B_{\text{OCV}}\text{OCV}(z) + B_II(t),
\end{align}
with $ A= A_{11}-A_{12}{A_{22}}^{-1}A_{21}$ with $A_{12}{A_{22}}^{-1}A_{21} = \bar{B}\Pi_{v}W$, $\,B_{\text{OCV}} = \bar{B}\Pi_v$ and $B_I= \bar{B}\Pi_I$ where $ \bar{B}=A_{12}$. 

Two key features of this  model are 1. it is an ODE whose vector field is written explicitly in terms of the circuit parameters (the various resistance and capacitances) and 2. the model nonlinearities (from the open circuit voltages OCVs) enter in an affine manner. With the added assumption that these OCVs are slope-restricted, then this nonlinear circuit model can be thought of as a Lurie system \cite{khalil2002nonlinear,drummond2019feedback,vidyasagar2002nonlinear,brogliato2007dissipative}, a class of nonlinear systems whose analysis is tractable (as illustrated in the observer design of the following section). %It is also highlighted that the current going into each cell is proportional to the relative resistance of each branch. One of the drawbacks of the model in its current form is that it is not clear how the response of each cell could be decoupled from each other. Decoupled dynamics would be advantageous for scalable algorithm design, but an analytic decoupling of the model equations does not appear obvious.

%The parallel connected pack's voltage can also be written in terms of the differential variables of the system.

\begin{remark}
{Since the main issue of obtaining the state-space form of the parallel pack model was resolving Kirchhoff's laws, it would seem that the above approach can be readily generalised to the case when the cell dynamics are described by electrochemical models like the single particle model \cite{guo2010single}, the single particle model with electrolyte and the Doyle-Fuller-Newman model \cite{doyle1993modeling,drummond2019feedback}.}
\end{remark}

%\begin{figure*}\begin{align}
%B_\text{OCV}
% = 
%-\begin{bmatrix} \bar{B}_1\sum_{i = 1}^{n-1}m_{1,i} &  -\bar{B}_1m_{11} &  -\bar{B}_1m_{12} & \dots &  -\bar{B}_1m_{1n-1}
%\\
%\vdots & \vdots & \vdots & \vdots &  \vdots 
%\\ 
% \bar{B}_{n}\sum_{i = 1}^{n-1}m_{1,n} & -\bar{B}_{n}m_{n1} & -\bar{B}_{n}m_{n2} & \dots & -\bar{B}_{n}m_{n,n-1}\end{bmatrix},
%\quad 
%B_\text{I}  &= -
% \begin{bmatrix}\bar{B}_1m_{1n} \\
% \vdots 
%\\ \bar{B}_{n-1}m_{n-1n} 
%\\
% \bar{B}_nm_{nn}
% \end{bmatrix}.\label{Pi_I}
%\end{align}
%\hrulefill
%\end{figure*}

\section{State estimator design}\label{sec:obs}
To illustrate the potential of this model, a state-estimator for a pack with parallel connected cells is now introduced. The key point of this estimator is that it guarantees boundedness of the estimation error of the parallel connected pack model's state to some set $\mathcal{E}$ (defined in Proposition \ref{theorem}) even when the nonlinear model is subject to disturbances. Once again, this result was strongly motivated by simplifying the convergence criteria of \cite{D_L}. 

%\begin{figure*}
%\begin{align}\label{LMI}
%\Omega = \begin{bmatrix}PA + A'P + QC + C'Q'-\overline{\delta}\underline{\delta} Z'\tau Z&  PB_{\text{OCV}}+ QD_{\text{OCV}}+\frac{\overline{\delta}+\underline{\delta}}{2} Z'\tau
%& PB_{d_x} & QD_{d_v}
%\\ 
% \star &-\tau &  0 & 0  \\ 
%\star &  \star & -1 & 0 \\ 
%\star & \star & \star & -1   \end{bmatrix} \prec 0
%\end{align}
% \hrulefill
%\end{figure*}

\subsection{State-estimator}

The goal of the estimator will be to obtain more accurate values of the states under the assumption that the pack is being perturbed by disturbances on the current $d_I \in \mathcal{L}_2$ and voltage $d_v \in \mathcal{L}_2$. Under this assumption, the battery model plant becomes
\begin{subequations}\label{pertsyst}\begin{align}
\dot{x}(t)  & = Ax(t) + B_{\text{OCV}}\text{OCV}(z(t)) + B_I(I(t) + d_I(t)), \\
\check{v}(t)   &  = Cx(t) + D_{\text{OCV}}\text{OCV}(z) +D_I(I(t) + d_I(t)) + d_v(t).
\end{align}\end{subequations}

The following state-estimator is proposed for this system
\begin{align}\label{sys_obs}
\dot{\hat{x}}(t) = A \hat{x}(t) + B_{\text{OCV}}\text{OCV}(\hat{z}(t))+B_II(t) - K(\check{v}(t)-\check{\hat{v}}(t)) 
\end{align}
with $\hat{x} = [\hat{z}_k , \hat{w}_k], \, k = 1, \dots ,\, n$ being the estimated states, $\check{\hat{v}}$ the predicted voltage, $\hat{i}_k$ the estimated currents and 
\begin{align}
\check{\hat{v}}(t) = \begin{bmatrix}\hat{v}(t) \\ \vdots \\ \hat{v}(t) \end{bmatrix} = \begin{bmatrix}\hat{w}_1(t) + \text{OCV}(\hat{z}_1(t)) + r_1\hat{i}_1(t)  \\ \vdots \\ \hat{w}_n(t) + \text{OCV}(\hat{z}_n(t))+ r_n\hat{i}_n(t) \end{bmatrix}
\end{align}
the voltage concatenation.
 
Defining the error between the plant \eqref{pertsyst} and the state estimator \eqref{sys_obs} as $e = x- \hat{x} $ then a set of error dynamics can be written
 \begin{align}\label{sys_err}
\dot{e}(t) = A {e}(t) +B_{\text{OCV}}\Delta \text{OCV} +  K (\check{{v}}(t)-\check{\hat{v}}(t)) + B_Id_I(t),
\end{align}
where $\Delta \text{OCV} $ is the open circuit voltage error
\begin{align}
\Delta \text{OCV} = \text{OCV}({z}(t))-\text{OCV}(\hat{z}(t)).
\end{align}
Demonstrating boundedness of this error system guarantees  that the estimator can provide a good estimate of the plant's state, even when it is subject to the disturbances $d_I(t)$ and $d_v(t)$.

\subsection{Estimator design}

The following proposition can be used to guide the design of this state estimator.

\begin{proposition}\label{theorem}
Consider the Li-ion battery pack model \eqref{pertsyst} with the state estimator \eqref{sys_obs} and assume that each cell's open circuit voltage OCV($z_k$) is a strictly monotonic function of its state-of-charge $z_k$
\begin{align}\label{OCV}
\frac{d\text{OCV}(z_k)}{d z_k} = \delta \in [\underline{\delta}, \overline{\delta}], \quad \underline{\delta} >0.
\end{align}
Set the estimator gain to $K = \overline{K} + \tilde{K}$ where $\overline{K} =  -\bar{B}\Pi_v(I_n + \bm{r}\Pi_v)^{-1}$, $\tilde{K} = \kappa(I_n + \bm{r}\Pi_v)^{-1}$ and $\kappa = \text{diag}(\kappa_1, \dots, \kappa_n)$ is built from the $2 \times 1$ blocks
 \begin{align}
 \kappa_j = \begin{bmatrix} \kappa^1_j \\ \kappa_j^2\end{bmatrix}.
\end{align} 
The gains $\kappa_j^1, \, \kappa_j^2$ are chosen such that the roots of the quadratic
\begin{align}\label{p_s}
    p(s,\delta)  & =s^2-b(\delta)s+ c(\delta),
\end{align}
with 
\begin{align}
    b(\delta) &  = \left(\kappa_j^1\delta+\kappa_j^2-\frac{1}{R_jC_j}\right), 
    \\  c(\delta)  &=\kappa^1_j\delta\left(\kappa_j^2-\frac{1}{R_jC_j}-  \kappa_j^2 \right), 
\end{align}
strictly lie within the left half plane for all $\delta \in [\underline{\delta}, \, \overline{\delta}]$ and each $j = 1, \, \dots, \, n$. With this choice of gain, then the error system is stable in the large meaning that is has a unique attractive global equilibrium point.

\end{proposition} 

\begin{proof}
With the chosen gain $K$, the error dynamics can be expanded out as 
 \begin{subequations}\label{error_sys_proof}\begin{align}
\dot{e}(t)  & = A {e}(t) +B_{\text{OCV}}\Delta \text{OCV} +  K (\check{{v}}(t)-\check{\hat{v}}(t)) + B_Id_I(t),
\\
 & = 
 (A_{11}+ \bar{B}\Pi_{v}W +  (\overline{K}+\tilde{K}) (I_n + \bm{r}\Pi_{v})W) {e}(t) 
 \\
 & \quad+\left(\bar{B}\Pi_v + (\overline{K}+\tilde{K}) (I_n + \bm{r}\Pi_v)\right)\Delta \text{OCV} \nonumber
 \\ &\quad  +Kd_v(t) + B_Id_I(t). \nonumber
\end{align}\end{subequations}
Substituting in the expressions for $\overline{K}$ and $\tilde{K}$ and defining 
\begin{align}
    d(t) = Kd_v(t) + B_Id_I(t),
\end{align} then these error dynamics reduce to
\begin{align}
\dot{e}(t)  & = 
 (A_{11}+ \kappa W) {e}(t) +\kappa\Delta \text{OCV} +d(t).\label{dyns_e_nonlin}
\end{align}
The choice of gain has decoupled the error dynamics of each cell from one another, with \eqref{dyns_e_nonlin} being composed of $n$ decoupled, second order systems. Since the error dynamics of each cell are now decoupled and second order, they satisfy the Kalman conjecture \cite{KC1,KC2,KC3,KC4} and so the stability of the nonlinear system can be checked from its linearisation as conjectured (wrongly for the general case) by Kalman in \cite{kalman1957physical} with its modern interpretation given in \cite{drummond2018aizerman}. As such, for each of these decoupled second order dynamics representing each cell in the pack, verifying the stability of the nonlinear system \eqref{dyns_e_nonlin} is equivalent to checking the stability of each of its linearisations.

Linearising the error dynamics \eqref{dyns_e_nonlin} for each cell (with $e_k$ being the error in the state prediction of the $k^\text{th}$ cell) gives 
\begin{align}
\dot{e}_k = E_k(\delta)e_k
+  \begin{bmatrix}d_{2k-1} \\ d_{2k}\end{bmatrix},
\end{align} 
where 
\begin{align}\label{P}
E_k(\delta)  = \begin{bmatrix} \kappa_k^1 \delta  & \kappa_k^1 \\ \kappa_k^2 \delta & \kappa_k^2-\frac{1}{R_kC_k}\end{bmatrix}.
\end{align}
Each of these linear systems are stable provided $E_k(\delta)$ is Hurwitz for all $\delta \in [\underline{\delta}, \, \overline{\delta}]$. And, since \eqref{p_s} is the characteristic equation of $E_k(\delta) $, it's roots determine the eigenvalues of $E_k(\delta)$.

\end{proof}

\begin{remark}
Convergence of the state estimator error from Proposition \ref{theorem} implies that the system \eqref{pertsyst} is at least detectable. However, unless some specific feature in the system structure can be exploited, proving the stronger notion of observability for the nonlinear system   will prove challenging, as it will involve computing Lie derivatives along the vector field, which do not scale well to large systems.  
\end{remark}

\begin{remark}
Standard state estimators such as the extended Kalman filter (EKF) could also be applied but the design of Proposition \ref{theorem} has the benefit of a) exploiting the system structure to decouple the cell dynamics and b) providing a simple check to guarantee error convergence. Normally, the EKF is not accompanied with similar guarantees. 
\end{remark}

The main benefit of Proposition \ref{theorem} is that it gives algebraic conditions to construct the estimator gains \eqref{p_s} for the nonlinear system. These conditions are rather simple to check, but stronger results may be obtained using a numerical search. This is exemplified by the following proposition which applies the classical observer design approach of \cite{arcak2001nonlinear} to obtain an upper bound for the observer error and because the following is an linear matrix inequality it can be solved using convex optimisation.

\begin{proposition}\label{prop2}
Define the matrices
\begin{subequations}\begin{align}
    A_e  & = A_{11}+ \bar{B}\Pi_{v}W,
    \quad 
    A_{e,2}   = (I_n + \bm{r}\Pi_{v})W, \\
    B_e  & = \bar{B}\Pi_v ,
    \quad 
    B_{e,2}=  (I_n + \bm{r}\Pi_v).
\end{align}\end{subequations}
and, for some $P \in \mathbb{S}_{\succ 0}^{2n}$, $Q \in \mathbb{R}^{2n\times n} $, $\gamma \geq 0$ and $\tau \in \mathbb{D}^n_+$,
\begin{align}
    M = \begin{bmatrix} PA_e + QA_{e,2} & PB_e  +QB_{e,2} & P \\ 0 & 0 & 0 \\
    0 & 0 & -\frac{1}{2}\gamma\end{bmatrix} 
\end{align}
\begin{align}
    \Omega = \begin{bmatrix} -\underline{\delta}\overline{\delta}Z^TZ & \frac{1}{2} (\overline{\delta}+\underline{\delta})\tau Z & 0 \\ \frac{1}{2} (\overline{\delta}+\underline{\delta})Z^T\tau & -\tau & 0 \\ 0 & 0 & 0 \end{bmatrix},
\end{align}
where $Z \in \mathbb{R}^{2n \times n}$ is a matrix full of zeros apart from $Z_{k,j} = 1$ if $j = 2k-1$ for $k = 1, \, \dots, \, n$ (in other words $z = Zx$).

If the linear matrix inequality
\begin{align}
M + M^T + \Omega \prec 0 
    \end{align}
 holds, then  with the estimator's gain set to $K = P^{-1}Q$ its error is bounded from above by
\begin{align}
    V(e(t)) \leq V(e(0))  + \gamma \|d\|_2^2
\end{align}
for all $d \in \mathcal{L}_2$ where $V(e(t)) = e(t)^TPe(t)$ provided the state-of-charge of both the estimator and the plant remain within $[0,1]$. 
\end{proposition}
\begin{proof}
By writing the error dynamics of \eqref{error_sys_proof} as
 \begin{subequations}\begin{align}
\dot{e}(t)  &  = 
 (A_{11}+ \bar{B}\Pi_{v}W +  K(I_n + \bm{r}\Pi_{v})W) {e}(t) 
 \\
 & \quad+\left(\bar{B}\Pi_v + K (I_n + \bm{r}\Pi_v)\right)\Delta \text{OCV} 
+d(t), \nonumber
 \\
 & = (A_e + KA_{e,2})e(t) + (B_e+KB_{e,2}) \Delta \text{OCV} + d(t)\nonumber,
\end{align}\end{subequations}
these dynamics are in the standard Lurie system form which allows the classical nonlinear observer synthesis results of \cite{arcak2001nonlinear} to be applied. Applying these results gives the conditions and bounds of the theorem. 
\end{proof}

\begin{figure}
\centering
%\graphicspath{ {./Figures/} }
\includegraphics[width=0.4\textwidth]{{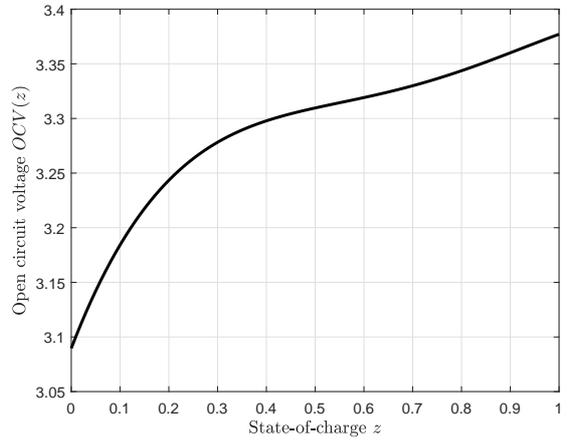}}
        \caption{{Open circuit voltage.}} 
        \label{fig:ocv}
\end{figure}

\begin{figure}
\centering
%\graphicspath{ {./Figures/} }
\includegraphics[width=0.4\textwidth]{{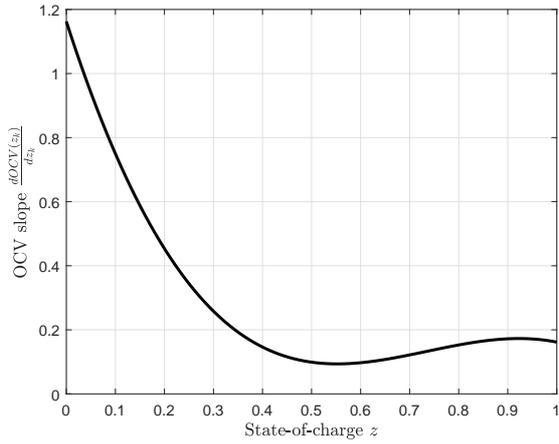}}
        \caption{{Slope of the open circuit voltage.}} 
        \label{fig:ocv_slope}
\end{figure}

\begin{figure}
\centering
%\graphicspath{ {./Figures/} }
\includegraphics[width=0.4\textwidth]{{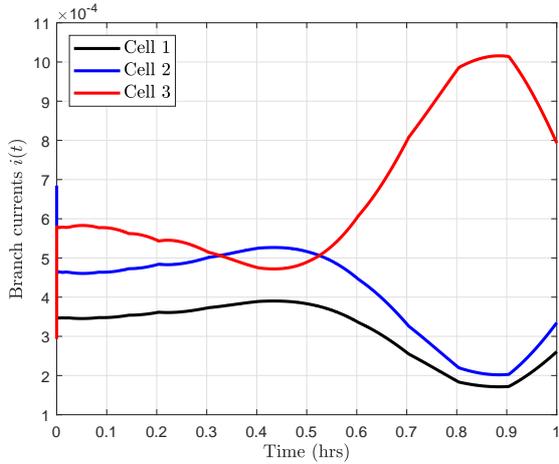}}
        \caption{{Branch currents during the charge.}} 
        \label{fig:i_branch}
\end{figure}

\begin{figure}
\begin{subfigure}{.5\textwidth}
\centering
%\graphicspath{ {./Figures/} }
\includegraphics[width=0.8\textwidth]{{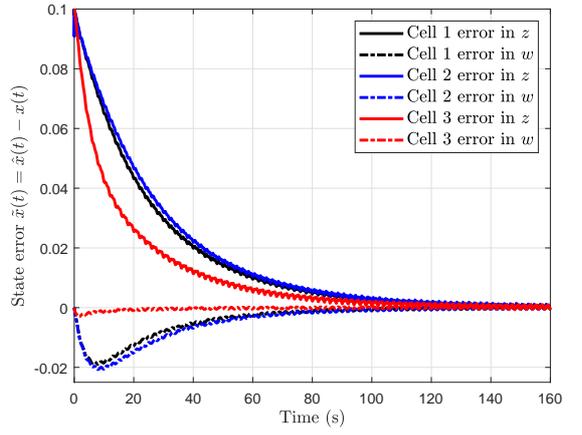}}
        \caption{{Error in the state estimates.}} 
        \label{fig:error_x}
\end{subfigure}
\begin{subfigure}{.5\textwidth}
\centering
%\graphicspath{ {./Figures/} }
\includegraphics[width=0.8\textwidth]{{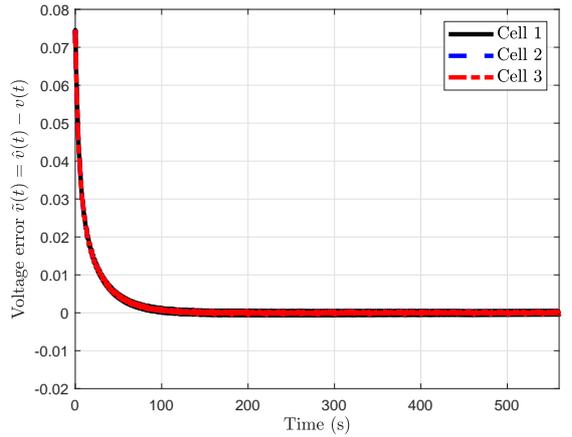}}
        \caption{{Error in the voltage.}} 
        \label{fig:error_v}
\end{subfigure}
\caption{Rapid convergence in both the state and voltage error of the state-estimator designed using Proposition \ref{theorem}.  }
\label{fig:error}
\end{figure}

\begin{table}
\centering 
\renewcommand{\arraystretch}{1.3} 
\begin{tabular}{|c|c|c|c|}
\hline 
& Cell 1 & Cell 2 & Cell 3 \\ 
\hline
$r_k$ & 0.0040& 0.0035 & 0.00045 \\ \hline
$R_k$ & 0.0025 & 0.0015 & 0.0035\\ \hline
$Q_k$ & 1.7 & 2 & 2.3 \\ \hline
$C_k$ & 1500 & 2000& 1000\\ \hline
\end{tabular}
\caption{Cell parameters for the numerical simulation.}
\label{tab:params}
\end{table}

% \begin{table}
%\centering 
%\renewcommand{\arraystretch}{1.3} 
%\begin{tabular}{c|c}
%\hline 
%\multicolumn{2}{c}{OCV coefficients} \\
%\hline
%Coefficient & Value \\ 
%\hline
%$a_0$ & 3.0896\\ 
%$a_1$ & 1.1627\\ 
%$a_2$ & -2.3821 \\ 
%$a_3$ & 2.1870 \\ 
%$a_4$ & -0.5444 \\ 
%$a_5$ & -0.1939\\ 
%$a_6$ &  0.0582\\  \hline
%\end{tabular}
%\caption{Nomenclature of the pouch cell model.}
%\label{tab:nom_pouch}
%\end{table}

\begin{table}
\centering 
\renewcommand{\arraystretch}{1.0} 
\begin{tabular}{|c|c |c|c|c|c|c|}
\hline
$a_0$  & $a_1$ & $a_2$& $a_3$& $a_4$ & $a_5$ & $a_6$ \\ \hline
 3.0896 & 1.1627 & -2.3821 & 2.1870 & -0.5444 & -0.1939 &  0.0582 \\  \hline
\end{tabular}
\caption{Coefficients of the OCV curve in \eqref{OCV}.}
\label{tab:ocv}
\end{table}

\section{Simulations}

A simulation is introduced in this section to illustrate the potential of the proposed state-estimator and ODE model for the parallel connected Li-ion batteries. Consider three NMC cells connected in parallel with parameter values given in Table \ref{tab:params} taken from \cite{D_L} and the open circuit voltage 
\begin{align}\label{OCV}
\text{OCV}(z) = \sum_{k = 0}^6 a_kz^k
\end{align}
with coefficients $a_k$ given in Table \ref{tab:ocv}. This OCV is shown in Fig. \ref{fig:ocv} and its slope is given in Fig. \ref{fig:ocv_slope}, clearly indicating its strict monotonicity. From Figure \ref{fig:ocv_slope}, the upper and lower slope bounds $\underline{\delta}= 0.0936$ and $\overline{\delta} = 1.1627$  for the OCV are obtained.

Figure \ref{fig:i_branch} shows a simulation of the branch currents of this parallel connected pack under a 1C charging current with $I(t) = 1.4 \times 10^{-3}~A$. The initial conditions were such that the initial state-of-charge for each cell was $z_1(0) = 0.05$, $z_2(0) = 0.1$, $z_3(0) = 0.15$ and the relaxation voltages were zero with $w_k = 0$ for $k = 1,\,2,\,3$. The non-uniform branch currents of the pack are clearly visible in this simulation. 

Figure \ref{fig:error} examines the performance of the state estimator designed in Proposition \ref{theorem}. For this simulation, the plant was again charged at 1C  from the same initial conditions. The observer state was initialised by $\hat{z}(0) = z(0)-0.05$, $\hat{w} = 0$ and the gain was set to
\begin{align}
\kappa_k = -\begin{bmatrix} 0.1 \\  0.1 \end{bmatrix}, \quad \forall k = 1, \, \dots, \, n = 3.
\end{align}
This choice of gain satisfies the stability conditions of Proposition \ref{theorem}. Sinusoidal disturbances were assumed to be perturbing the current and voltage measurements with
\begin{subequations}\begin{align}
d_I(t)  & = I(t)\sin(2\pi t),
\\
d_v(t)  & = I(t)\sin(\pi t). 
\end{align}\end{subequations}

The convergence of the voltage and state errors with this estimator design is shown in Fig. \ref{fig:error}, justifying  the claims of Proposition \ref{theorem}.

\section*{Conclusions}
%In this paper, we introduced a new model for a Li-ion batter pack where the cells are connected in series. The main result was to provide an analytical expresion for the algebraic equation which allows the currents going into the cells to be expressed in terms of the voltages. With this analytic expression in hand, we could eliminate the algebraic equations from the pack model and turn it into an ODE. The benefits of this are nto just computational speed-up but it also provides insight that can help battery management systems for these packs. To illustrate this, a state estimator for this pack was proposed that could guarnateed exponential convergence for eveeyr state of the pack, even for the nonlinear system. We hope that that work can stimulate further ideas on...

This paper has introduced a state-space model for lithium ion battery packs connected in parallel. The key result was the solution to Kirchhoff's laws for parallel connected packs, where the various branch currents charging each cell could then be written explicitly in terms of the pack's state-space variables, applied current and the various cell resistances. In this way, the model avoids the need to compute these branch currents numerically. The analytic solution for the branch currents brings insight into the model's dynamics and structure, as highlighted in this paper by the design of a new state estimator for the nonlinear pack model. Simple conditions are stated for this estimator's gains that guarantee its error is convergent, with the conditions being derived from an application of Aizerman's conjecture. It is hoped that the methods developed in this work will lead to the transfer of ideas from the model simulation and battery management system design of series connected battery packs to their parallel counter-parts.

\bibliographystyle{IEEEtranS}
\bibliography{bibliog}

\end{document}